\documentclass{amsart}

\usepackage{latexsym,amsfonts} 
\usepackage{amsmath,amssymb,graphics,setspace}
\usepackage{amsthm}
\usepackage{MnSymbol}
\usepackage{mathrsfs} 
\usepackage{fontenc}%

\usepackage{paralist}
\usepackage{color}
\usepackage[verbose]{wrapfig}

\usepackage[ruled,vlined,linesnumbered]{algorithm2e}

\usepackage{marvosym}

\usepackage{picins,graphicx}

\usepackage{pstricks,pst-text,pst-grad,pst-node,pst-3dplot,pstricks-add,pst-poly,pst-coil} 
\usepackage{multido,pst-node,pst-bspline}

\usepackage{subfigure}
\renewcommand{\thesubfigure}{\thefigure.\arabic{subfigure}}
\makeatletter
\renewcommand{\p@subfigure}{}
\renewcommand{\@thesubfigure}{\thesubfigure:\hskip\subfiglabelskip}
\makeatother

\usepackage{hyperref}
\hypersetup{linktocpage=true,colorlinks=true,linkcolor=blue,citecolor=blue,pdfstartview={XYZ 1000 1000 1}}

\newcommand{\cl}{\mbox{cl}}
\newcommand{\Int}{\mbox{int}}
\newcommand{\bdy}{\mbox{bdy}}
\newcommand{\fil}{\mbox{fil}}
\newcommand{\Nrv}{\mbox{Nrv}}
\newcommand{\near}{\delta} 
\newcommand{\dnear}{\delta_{\Phi}} 
\newcommand{\assign}{\mathrel{\mathop :}=}

\newcommand{\sk}{\mbox{sk}}
\newcommand{\cx}{\mbox{cx}}
\newcommand{\shape}{\mbox{sh}}

\newcommand{\maxNrvClu}{\mbox{maxNrvClu}}

\newcommand{\sn}{\mathop{\delta}\limits^{\doublewedge}} 

\newtheorem{example}{Example}
\newtheorem{remark}{Remark}
\newtheorem{definition}{Definition}
\newtheorem{lemma}{Lemma}
\newtheorem{theorem}{Theorem}

\begin{document}

\title[Delta Complexes]{Delta Complexes in Digital Images.\\  Approximating Image Object Shapes}

\author[M.Z. Ahmad]{M.Z. Ahmad$^{\alpha}$}
\email{ahmadmz@myumanitoba.ca}
\address{\llap{$^{\alpha}$\,}
Computational Intelligence Laboratory,
University of Manitoba, WPG, MB, R3T 5V6, Canada}

\author[J.F. Peters]{J.F. Peters$^{\beta}$}
\email{James.Peters3@umanitoba.ca}
\address{\llap{$^{\beta}$\,}
Computational Intelligence Laboratory,
University of Manitoba, WPG, MB, R3T 5V6, Canada and
Department of Mathematics, Faculty of Arts and Sciences, Ad\.{i}yaman University, 02040 Ad\.{i}yaman, Turkey}
\thanks{The research has been supported by the Natural Sciences \&
Engineering Research Council of Canada (NSERC) discovery grant 185986 
and Instituto Nazionale di Alta Matematica (INdAM) Francesco Severi, Gruppo Nazionale per le Strutture Algebriche, Geometriche e Loro Applicazioni grant 9 920160 000362, n.prot U 2016/000036.}


\date{}

\dedicatory{Dedicated to P. Alexandroff and Som Naimpally}

\begin{abstract}
In a computational topology of digital images, simplexes are replaced by Delta sets in approximating image object shapes.  For simplicity, simplexes and Delta sets are restricted to the Euclidean plane.  A planar simplex is either a vertex, a line segment or a filled triangle.  In this study of image shapes, a planar Delta set is a sequence of ordered simplicial complexes.  The basic approach is to approximate an image shape by decomposing an image region containing the shape into combinations of Delta sets called Delta complexes.  This approach to image shapes is motivated by the ease with which shapes covered by Delta complexes can be measured and compared.  A number of basic results directly related to shape analysis are also given in the context of Delta complex proximities. 
\end{abstract}

\maketitle

\section{Introduction}
A conventional approach to the computational topology~\cite{Peters2013springer} and proximity~\cite{Peters2016CP} for digital images is based on the assumption that an image lies on a finite region of the euclidean plane. This approach only considers pixel values when computing the generating points used in triangulations based only on the pixel locations. Thus, a lot of information in the underlying of the geometry of images is hidden by this abstraction.  Moreover, a conventional decomposition of planar images leads to triangles with straight edges which are not well-suited to the geometry of image object shapes that typically have a mixture of round as well straight edges. It is obvious that detecting image shapes with many curved singularities (edges), has been of particular interest in the literature~\cite{grohs2013alpha}.  A natural extension to existing literature in the field of computational topology and proximity is to consider finite planar images on a manifold defined by pixel values. This richer notion of an image will lead to triangulation of the manifold and will lead to triangles with curved edges with curvature derived from pixel values. Such triangulation schemes have been proposed in literature, which are based on Riemannian metrics~\cite{rouxel2016discretized}. This article also presents a method of curved triangulations based on B-Splines.

The notion of a simplicial complex is well adapted to triangulations with straight edges, and have been used in literature to formalize a computational topology based framework for shape analysis in digital images~\cite{peters2017proximal}. The motivation behind this is to effectively model and study the shapes of objects found in nature. Objects in nature have curved boundaries, thus triangles with straight edges are not suited to model such objects. Thus, we extend the notion of triangulations to triangles with curved edges. To extend the framework developed earlier to this generalized setting, we need to replace the simplex with the notion of a Delta($\Delta$)-set. The notion of Delta($\Delta$)-sets has been introduced in literature~\cite{friedman2012survey},\cite{hatcher2002algebraic}. This article introduces a framework for computational topology and proximity similar to the one proposed in~\cite{peters2017proximal}. This framework is independent of the method used to generate the triangulations and is applicable to triangles with either curved or straight lines, generalizing the existing work.

This article is organized as follows. Section $2$ includes all the basic definitions that are used to build the framework. Section $3$ introduces the methods used in the study of shape geometry and curved triangulation, based on B-splines that are generalizations of Bezier curves.  Section 4 briefly presents applications of the proposed approach in shape analysis.  The results for proximal complexes are given in Section 5.

\section{Preliminaries}
In geometry, the notion of a simplex is the generalization of a triangle to higher dimensions.
  
\begin{definition}\label{def:simplex}
A $k$-simplex is a $k$-dimensional polytope $A$ (denoted by $\sk A$), which is the convex hull of the set of its $k+1$ vertices $v_1,\cdots,v_{k+1} \in \mathbb{R}^{k+1}$, {\em i.e.},
\[ 
\sk A = \left\{\theta_1v_1 + \cdots + \theta_{k+1}v_{k+1}:\sum_{i=1}^{k+1} \theta_i =1, \theta_i \geq 0 \forall i \right\}.
\]
\qquad \textcolor{blue}{\Squaresteel}
\end{definition}


\begin{wrapfigure}[10]{R}{0.40\textwidth}
\begin{minipage}{5 cm}
\centering
\begin{pspicture}(-0.75,-0.75)(4.2,1.7)
\psframe[linecolor=black](-0.5,-0.5)(4,1.5)
\psdots[dotstyle=o,dotsize=0.1,linecolor=red,fillstyle=solid,fillcolor=black](0,0)
\psdots[dotstyle=o,dotsize=0.1,linecolor=blue,fillstyle=solid,fillcolor=black](1.25,0)(1.75,0.5)
\psline[linewidth=1.5pt](1.25,0)(1.75,0.5)
\psdots[dotstyle=o,dotsize=0.1,linecolor=blue,fillstyle=solid,fillcolor=black](2.5,0)(3,1)(3.5,0)
\pspolygon[linewidth=1.5pt,fillstyle=solid,fillcolor=gray](2.5,0)(3,1)(3.5,0)
\rput(0,1.2){$k=0$}
\rput(1.25,1.2){$k=1$}
\rput(2.5,1.2){$k=2$}
\end{pspicture}
\caption[]{$k$-simplices}
\label{fig:simplex}
\end{minipage}
\end{wrapfigure}
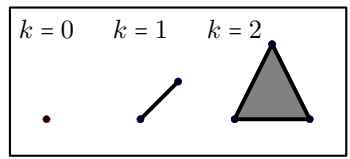

A $0$-simplex is a point, $1$-simplex is a line, a $2$-simplex is a triangle and so on. This is displayed in figure~\ref{fig:simplex} To analyze the topology of planar digital images we require $k$-simplices with $k \leq 2$. A $k$-simplex can be uniquely defined by the set of its vertices. Thus we can replace the notion of the geometrical $k$-simplex as defined in the definition~\ref{def:simplex} with the notion of an abstract simplex which is defined by its vertices. This notion of an abstract simplex has the same combinatorial information as the geometrical $k$-simplex and is homeomorphic to it.

The fact that simplex has straight lines is evident, both from the fact that it is a polytope and its definition as a linear combination of it vertices (from definition~\ref{def:simplex}). Using the simplices as building blocks, we can construct elaborate topological structures known as simplicial complexes. A simplicial complex $X \in \mathbb{R}^{k+1}$ is a collection of simplices of various dimensions in $\mathbb{R}^{k+1}$ such that every face of a simplex of $X$, is in $X$ and the intersection of any two simplices of $X$ is a face of each of them. A face of a $k$-simplex (defined by $k+1$ vertices ) is a subset of the vertex set containing $k$ vertices, which is itself a simplex.

In geometry and topology, a \emph{simplex} is the convex hull (smallest closed convex set) which contains a nonempty set of vertices~\cite[\S I.5, p. 6]{Alexandroff1932elementaryConcepts}. A \emph{convex hull} of a set of vertices $A$ (denoted by $\mbox{\textbf{conv} A}$) in $n$ dimensions is the intersection of all convex sets containing $\mbox{\textbf{conv} A}$~\cite{Ziegler2207polytopes}. A \emph{convex set} $A$, contains all points on each line segment drawn between any pair of points contained in $A$.   For example, a $0$-simplex is a vertex, a $1$-simplex is a straight line segment and a $2$-simplex is a rectilinear triangle that includes the plane region which it bounds (see Fig. 1). In general, a $k$-simplex is a polytope with $k+1$ vertices. 
A \emph{polytope} is an $n$-dimensional point set $P$, that is an intersection of finitely many closed half spaces in the Euclidean space $\mathbb{R}^n$.
So, for example, a $2$-simplex is $2$-dimensional polytope (i.e., a filled triangle) in the Euclidean plane represented by $\mathbb{R}^2$. 
A collection $X$ of simplexes is called a \emph{simplicial complex} (briefly, \emph{complex}).  
In this report, the study of complexes is restricted to the Euclidean plane.
A 2-dimensional \emph{face} is the set of points in the interior of a planar simplex.  For example, the face of 1-simplex is the set of points in the interior of a straight line segment and the face of a filled rectilinear triangle is the set of points in the interior of the triangle.

\begin{wrapfigure}{R}{0.50\textwidth}
\begin{minipage}{6 cm}
\centering
\begin{pspicture}(-0.75,-0.75)(4.25, 4.25)
\psframe[linecolor=black](-0.5,-0.5)(4.2,4.2)
\psdots[dotstyle=o,dotsize=0.15,fillstyle=solid,fillcolor=red](0,1.5)(0.5,1)(0.75,2)
\pspolygon[linewidth=1.5pt,fillstyle=vlines,fillcolor=gray](0,1.5)(0.5,1)(0.75,2)
\psdots[dotstyle=o,dotsize=0.15,fillstyle=solid,fillcolor=red](2,1.5)(2.5,2.5)(3,1.5)(2.5,0.5)(2.5,2.5)(3.15,2.15)(3.1,3.15)
\pspolygon[linewidth=1.5pt,fillstyle=solid,fillcolor=gray](2,1.5)(2.5,2.5)(3,1.5)
\pspolygon[linewidth=1.5pt,fillstyle=solid,fillcolor=gray](2,1.5)(2.5,0.5)(3,1.5)
\pspolygon[linewidth=1.5pt,fillstyle=vlines,fillcolor=gray](2.5,2.5)(3.15,2.15)(3.1,3.15)
\pscurve{->}(1,2.1)(2,3.25)(2.5,3)
\rput(1.5,3.5){Inclusion Map}
\rput(-0.2,1.3){0}
\rput(0.6,0.8){1}
\rput(0.7,2.25){2}
\rput(2.35,2.7){0}
\rput(3.3,2){1}
\rput(3.3,3,3){2}
\end{pspicture}
\caption[]{Simplicial Complexes and inclusion maps}
\label{fig:simplicial_complex}
\end{minipage}
\end{wrapfigure}
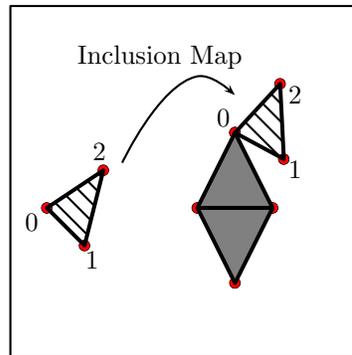

In figure~\ref{fig:simplex} each of the lines in the $2$-simplex is a face and so is each of the endpoints in the $1$-simplex. Before moving on to the simplicial complex, we define the notion of an inclusion map. Given two sets $B \subseteq A$, a map $f:B \rightarrow A$ which satisfies $f(b)=b$ for all $b \in B$, is called an inclusion map. A more formal definition of an abstract simplicial complex is given below. 

In the sequel, we will consider curvilinear 2-simplexes, which are projections from Alexandroff rectilinear triangles onto curvilinear (round) triangles on the surface of a sphere.   Such triangles can be obtained by positioning a tetrahedron $T$ inside a sphere, so that the center of $T$ coincides with the center of the sphere and the vertices of $T$ lie on the surface of the sphere.   By projecting the edges of the inscribed tetrahedron onto the surface of the sphere, we obtain curvilinear triangles which are the rounded tetrahedron faces inscribed onto the sphere surface~\cite[\S 1.1, p. 7]{Fomenko1994SpringerVisualGeometry}.

\begin{definition}\label{def:simplicial_complex}\cite{friedman2012survey}
An abstract simplicial complex consists of "vertices" $X^0$ together with, for each integer k, a set $X^k$ consisting of subsets (not necessarily all of them) of cardinality $k+1$. These must satisfy the condition that any $(j+1)$-element subset of an element of $X^k$ is an element of $X^j$.
\qquad \textcolor{blue}{\Squaresteel}
\end{definition}
%

\begin{example}
Let us consider the shaded simplicial complex in the figure~\ref{fig:simplicial_complex}. In light of the definition~\ref{def:simplicial_complex} we can specify this simplicial complex in the plane as $\{X^0,X^1,X^2\}$. Where $X^0=\{0,1,2\}$, $X^1=\{\{0,1\},\{1,2\},\{0,2\}\}$ and $X^2=\{\{0,1,2\}\}$.
\qquad \textcolor{blue}{\Squaresteel}
\end{example}
The last condition in the definition ensures that every face of an abstract simplex in the abstract simplicial complex is also a simplex of the complex. An example of a simplicial complex is illustrated in figure~\ref{fig:simplicial_complex}, where both the simplexes colored gray and the shaded simplex are examples of abstract simplicial complexes. Moreover, the combination of the two, under the inclusion map is also a simplicial complex. 
%
To formalize this notion of combination of two simplices, we will introduce the notion of an inclusion map.An inclusion map is a morphism between two simplicial complexes. This map includes a simplex into a larger simplex. An example is shown in figure~\ref{fig:simplicial_complex}, where the shaded simplicial complex is included into the combined simplex, containing both the shaded and the filled simplicial complexes. Inclusion maps are a special case of the simplicial maps, which have been discussed in detail in classical texts\cite{munkres1984elements}. In this report we will use a notion of the $sew$ operator defined in\cite{Peters2016arXivProximalPhysicalGeometry},which is equivalent to the inclusion map.
For the simplicity of definition we assume that the simplicial complexes have faces which are at-most $2$-simplices. Let us define the $sew$ operator formally.

\begin{definition}\label{def:sew_operator}$\mbox{}$\\
Suppose $\cx A$ and $\cx B$ are two simplicial complexes with vertex sets $\{a_1,\cdots,a_p,\cdots,a_m\}$, $\{b_1,\cdots,b_q,\cdots,b_n\}, m,n \in \mathbb{Z}$ respectively. The complete specification as per definition~\ref{def:simplicial_complex} is $cx A= \{X_A^0=\{a_1,\cdots,a_m\},X_A^1,X_A^2\}$ and $cx B=\{X_B^0=\{b_1,\cdots,b_m\},X_B^1,X_B^2\}$, where each of the $X_A^i$ and $X_B^j$ is a planar region. A mapping, $sew: \mathbb{R}^2 \times \mathbb{R}^2 \longrightarrow \mathbb{R}^2$ yields a new simplicial complex $cx C$ defined as follows:
\[
\cx C\assign sew_{pq}\{\cx A,\cx B\}=\{X_C^i=X_A^i\cup X_B^i: a_p=b_q, i=0,1,2\}.
\]
Moreover, if for any subset $x_j$of $X_C^0$ and $|x_j|=j$, all possible combinations $\dbinom{j}{j-1}$ of elements in $x_j$ are in $X_C^{j-2}$ then, add $x_j$ to $X_C^{j-1}$.
\qquad \textcolor{blue}{\Squaresteel}
\end{definition}
\begin{example}
In figure~\ref{fig:sew2a} there are two planar simplicial complexes $cx A$ and $cx B$ which are defined as:
\begin{align*}
\cx A=&\{\{0,1,2\},\{\{0,1\},\{0,2\},\{1,2\}\},\{\{0,1,2\}\}\}. \\
\cx B=&\{\{a,b,c\},\{\{a,b\},\{a,c\},\{b,c\}\},\{\{a,b,c\}\}\}.
\end{align*}
The inclusion maps that define the operation shown in figure~\ref{fig:sew2a} are 
$f_{incB}:\cx B \rightarrow \cx C$ and $f_{incA}:\cx A \rightarrow \cx C$.
Now applying the sewing mapping we get:
\[
\cx C \assign sew_{2a}\{\cx A,\cx B\}=\{X_C^0,X_C^1,X_C^2\}, 
\]
where
\begin{align*}
X_C^0=&\{0,1,2\} \cup \{a=2,b,c\}=\{0,1,2,b,c\}.\\
X_C^1=& \{\{0,1\},\{0,2\},\{1,2\}\} \cup \{\{a=2,b\},\{a=2,c\},\{b,c\}\} \\
     =& \{\{0,1\},\{0,2\},\{1,2\},\{2,b\},\{2,c\},\{b,c\}\}. \\
X_C^2=&\{\{0,1,2\}\} \cup \{\{a=2,b,c\}\}=\{\{0,1,2\},\{2,b,c\}\}. \text{\qquad \textcolor{blue}{\Squaresteel}}.
\end{align*}
\end{example}
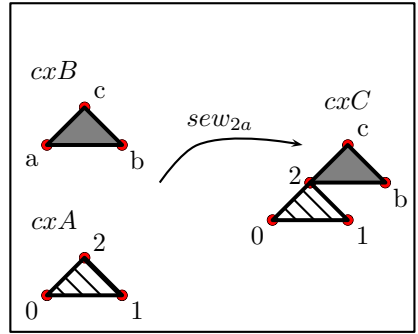
\begin{wrapfigure}{R}{0.50\textwidth}
\vspace{-10 pt}
\begin{minipage}{6 cm}
\centering
\begin{pspicture}(-0.75,-0.75)(5,5)
\psframe[linecolor=black](-0.5,-0.5)(4.9,3.9)
\psdots[dotstyle=o,dotsize=0.15,fillstyle=solid,fillcolor=red](0,0)(0.5,0.5)(1,0)
\pspolygon[linewidth=1.5pt,fillstyle=vlines,fillcolor=gray](0,0)(0.5,0.5)(1,0)
\psdots[dotstyle=o,dotsize=0.15,fillstyle=solid,fillcolor=red](0,2)(0.5,2.5)(1,2)
\pspolygon[linewidth=1.5pt,fillstyle=solid,fillcolor=gray](0,2)(0.5,2.5)(1,2)
\psdots[dotstyle=o,dotsize=0.15,fillstyle=solid,fillcolor=red](3.5,1.5)(4,2)(4.5,1.5)(3,1)(3.5,1.5)(4,1)
\pspolygon[linewidth=1.5pt,fillstyle=solid,fillcolor=gray](3.5,1.5)(4,2)(4.5,1.5)
\pspolygon[linewidth=1.5pt,fillstyle=vlines,fillcolor=gray](3,1)(3.5,1.5)(4,1)
\pscurve{->}(1.5,1.5)(2,2)(3.4,2)
\rput(-0.2,-0.2){0}
\rput(1.2,-0.2){1}
\rput(0.7,0.7){2}
\rput(0.1,1){$cxA$}
\rput(-0.2,1.8){a}
\rput(1.2,1.8){b}
\rput(0.7,2.7){c}
\rput(0.1,3){$cxB$}
\rput(2.8,0.8){0}
\rput(4.2,0.8){1}
\rput(3.3,1.6){2}
\rput(4.2,2.2){c}
\rput(4.7,1.3){b}
\rput(4,2.6){$cxC$}
\rput(2.3,2.3){$sew_{2a}$}
\end{pspicture}
\caption[]{Sewing two $2$-Simplices}
\label{fig:sew2a}
\end{minipage}
\end{wrapfigure}
The last condition in the definition handles the case, in which addition of a simplicial complex results in the completion of a higher order simplex, from the combination of lower order simplices. Suppose, that in figure~\ref{fig:sew2a} we add a $1$-simplex ($\overline{cd}$) to the $cx C$ by a composition of two $sew$ operators, $cx D = sew_{db} \circ sew_{1c}\{\bar{cd},cx C\}$. It is obvious that this will create a new $2$-simplex $\{1,2,b\}$,which is neither a face of the $cx C$ nor of the $1$-simplex $\bar{cd}$. As per the definition of the $sew$ operator (definition~\ref{def:sew_operator}), we can see that for a subset $\{1,2,b\}$ of $X_D^0$, all the possible combinations of $2$ elements (\{\{1,2\},\{1,b\},\{2,b\}\}) are in $X_D^1$. Thus we add this subset to $X_D^2$, because this subset defines a valid $2$-simplex of $cx D$. 

To impart some specificity to the framework defined above, let us consider the idea of an ordered simplicial complex.

\begin{definition}\label{def:orderedsmplx}
An ordered simplicial complex ($|\Delta^n|$) of order $n$, is defined by a totally ordered set of vertices $\{v_0,\cdots,v_i,v_j,\cdots,v_n\}$, provided $i < j$. 
\qquad \textcolor{blue}{\Squaresteel} 
\end{definition}

This leads to a slight change in the definitions of the simplicial maps, to make them order preserving so as to ensure that the result of the mapping is itself an ordered simplicial complex. We have discussed inclusion maps and equated them with the concept of a $sew$ operator. Now we define an important counter part of the inclusion map, namely the face map. Face maps are crucial in the specification of simplicial complexes as a combination of the constituent faces.
\begin{remark}
For a standard $n$-simplex specified by its vertex set $\{v_0,\cdots,v_n\}$, there are $n+1$ face maps $d_0,\cdots,d_n$. The maps are defined as $d_j:\{0,\cdots,j\cdots,n\} \rightarrow \{0,\cdots,n-1\}$ i.e. the result of face map $d_j$ is a $n-1$ face of the original $n$-simplex which is obtained by eliminating the $j-th$ vertex.
\end{remark}

The working of the face map can be understood by referring to the $smpxA$ in the figure~\ref{fig:sew2a}. Let us look at this with detail in the following example.

\begin{example}
The $2$-simplex $cxA$ in the figure~\ref{fig:sew2a}, can be specified by the vertex set $\{0,1,2\}$. There are $3$ face maps for this simplicial complex. These maps are defined as follows:
\begin{align*}
d_0:\{0,1,2\} &\rightarrow \{1,2\}.\\
d_1:\{0,1,2\} &\rightarrow \{0,2\}.\\
d_2:\{0,1,2\} &\rightarrow \{0,1\}.
\end{align*}
Thus the face maps of the $cxA$ yield the three $1$-simplices that are its faces. Extending this by looking at one of these faces e.g. the $1$-simplex defined by $\{1,2\}$. We can further define two face maps on this simplex:
\begin{align*}
d_0:\{1,2\} &\rightarrow \{1\}.\\
d_1:\{1,2\} &\rightarrow \{2\}.
\end{align*}
Thus a composition of face maps can be used to completely specify a simplicial complex from the $0$-simplices all the way up to the highest order simplex present. $d_0 \circ d_0(cxA)=\{1\}$ and $d_1 \circ d_0(cxA)=\{2\}$ are possible compositions of face maps. This leads to an arrow diagram that can specify the whole simplicial complex in terms of its constituent simplices as $\Delta^n \rightarrow \Delta^{n-1} \rightarrow \cdots \rightarrow \Delta^1 \rightarrow \Delta^0$.
\qquad \textcolor{blue}{\Squaresteel}
\end{example}

\begin{wrapfigure}{L}{0.50\textwidth}
\begin{center}
\includegraphics[width=15mm]{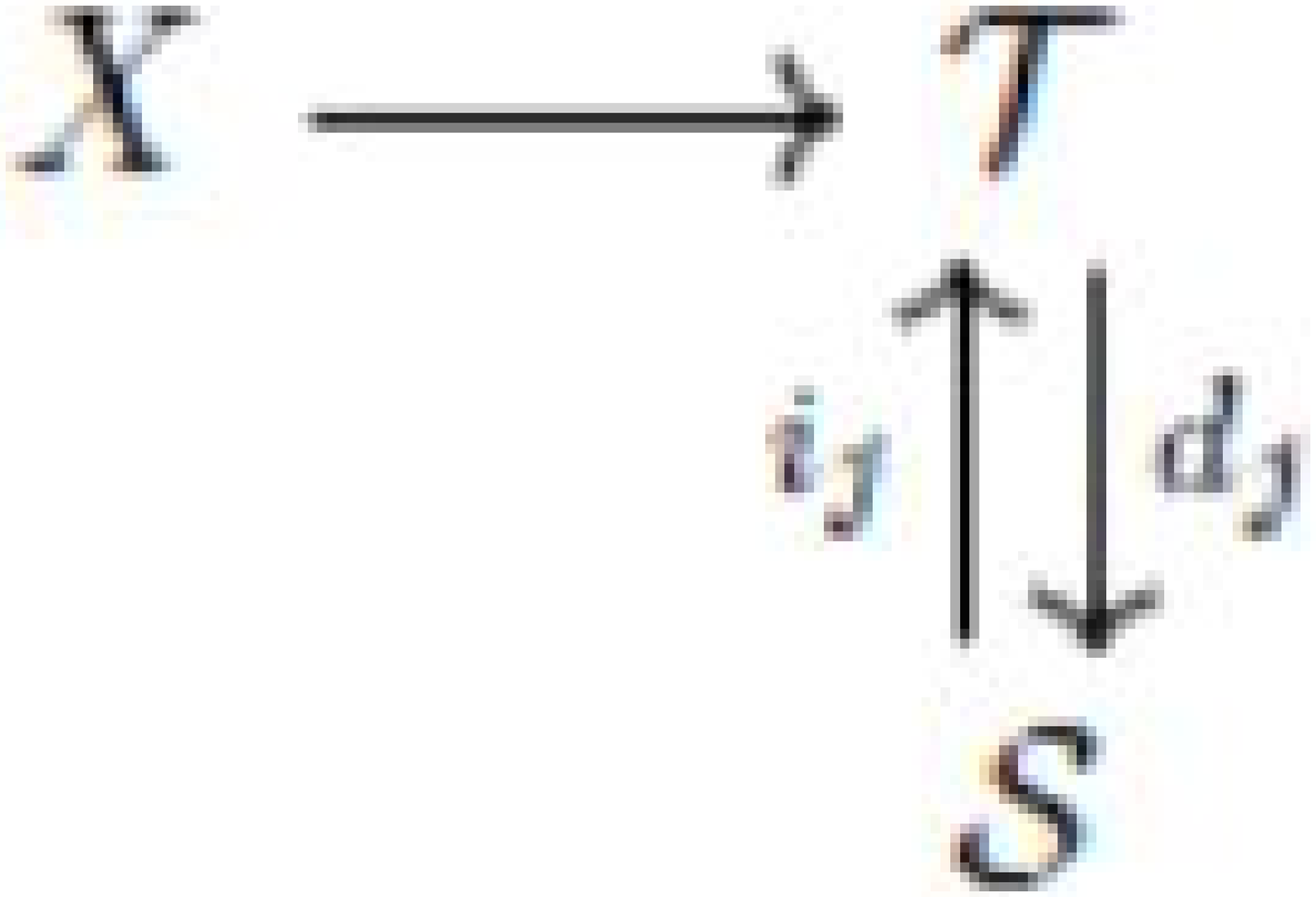}
\caption{Face maps and inclusion maps in a bundle framework}
\label{fig:bundle}
\end{center}
\end{wrapfigure}

Moreover, it is obvious from the definition of the face map that it is the inverse of an inclusion map. If $i_j: \mathcal{S} \rightarrow \mathcal{T}$ is an inclusion map that inserts a vertex into the simplicial complex $\mathcal{S}$ at location $j$ to yield the simplicial complex $\mathcal{T}$. Then, a face map $d_j:\mathcal{T} \rightarrow \mathcal{S}$ as it will remove the vertex at location $j$. This means that the face map $d_j$ is the inverse of the inclusion map $i_j$ as $d_j \circ i_j = id$, where $id$ is the identity map (see figure~\ref{fig:bundle}). 

Now, let us express a simplicial complex $\mathcal{T}$ as a result of the mapping from its vertex set $X$ and then specify this simplicial complex as the combination of its faces. Thus, the face map $d_j$ is the projection and the inclusion map $i_j$ is the section of the bundle.

As we have seen, that simplicial complexes can be specified as a combination of $0$-simplices using compositions of face maps and inclusion maps. Thus a simplicial complex can be defined, using an algebraic structure that specifies the $k$-simplices in the complex for all $k=0,\cdots,n$, present in the complex and the mappings that allow transition between these. The flexibility of the bundle structure allows for a compact representation as shown in figure ~\ref{fig:chain_bundle}. $X=\mathcal{T}^0$ is the vertex set of the simplicial complex and $\mathcal{T}^{k}$ is the set of all the $k$-simplices in the complex, where $k=0,\cdots,n$. Moreover, $d$ is the set of face maps that are the projections and the $i$ is the set of inclusion maps, which form the sections of this nested bundle structure at each level. 

\begin{wrapfigure}{R}{0.50\textwidth}
\begin{center}
\includegraphics[width=20mm]{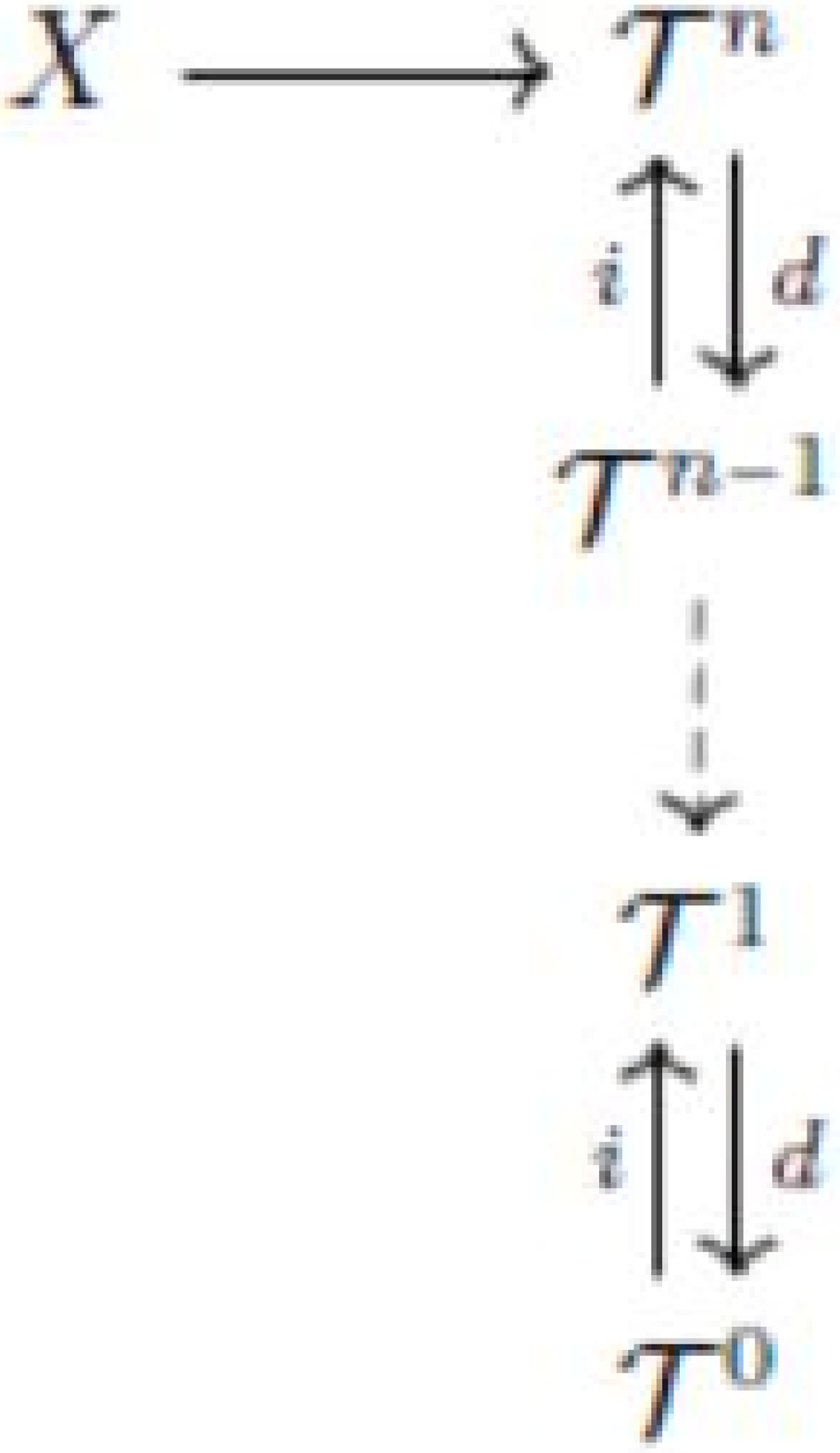}
\caption{Specification of a simplicial complex of order $n$ in terms of $0$-simplices}
\label{fig:chain_bundle}
\end{center}
\end{wrapfigure}

Before moving on to the definition of delta sets, we establish an important property of face maps for the case of ordered simplex. This property is stated as $d_id_j=d_{j-1}d_i$ for $i<j$. This can be easily verified by the following example.
\begin{remark}
Let $cxA$ be defined by vertex set $\{v_0,\cdots,v_i,\cdots,v_j,\cdots,v_n\}$, then it can be easily shown that $d_id_j=d_{j-1}d_i$.

\textbf{LHS:} $d_id_j[\{v_0,\cdots,v_i,\cdots,v_j,\cdots,v_n\}]=d_i[\{v_0,\cdots,v_i,\cdots,v_{n-1}\}]=\{v_0,\cdots,v_{n-2}\}$ 

\textbf{RHS:} $d_{j-1}d_i[\{v_0,\cdots,v_i,\cdots,v_j,\cdots,v_n\}]=d_{j-1}[\{v_0,\cdots,v_{j-1},\cdots,v_{n-1}\}]=\{v_0,\cdots,v_{n-2}\}$

Since, \textbf{LHS} $=$ \textbf{RHS}, hence it is proved that for an ordered simplicial complex $d_id_j=d_{j-1}d_i$ when $i < j$.
\qquad \textcolor{blue}{\Squaresteel}
\end{remark}
Let us move on to the definition of Delta($\Delta$)-sets, which can be seen as the generalization of ordered simplicial complexes.
\begin{definition}\cite{friedman2012survey}\label{def:delta_sets}
A Delta($\Delta$) set consists of a sequence of sets $X_0,X_1,\cdots$ an, for each $n \geq 0$, maps $d_i: X_{n+1} \rightarrow X_n$ for each $i$,$0 \leq i \leq n+1$, such that $d_id_j=d_{j-1}d_i$ whenever $i<j$. 
\end{definition}
A detailed analysis of the properties of Delta($\Delta$)-sets is presented in literature\cite{friedman2012survey}. It allows for such topological gluings which do not result in a simplicial complex. An example is to consider the $cx A$ in figure~\ref{fig:sew2a} and glue together vertex $0$ to vertex $1$ which yields a cone. Such a gluing does not yield a simplicial complex (in this case a $2$-simplex) as per the definition~\ref{def:simplicial_complex} because two of its vertices (vertex $0$ and vertex $1$) are the same.

It can be observed that the curvilinear triangles fall in the domain of non-euclidean geometries.The rectilinear triangles lie in the domain of euclidean geometries, where the sum of all the interior angles of the triangle is $180 \deg$. Curvilinear triangles fall into two geometries namely the the spherical and hyperbolic geometries. The sum of interior angles of a triangle in the case of spherical geometry is greater than $180^{\circ}$ and less than $180^{\circ}$ in hyperbolic geometry. Both forms of curvilinear triangles,spherical triangles (in blue) and hyperbolic triangles (in red), can be seen as the projections of rectilinear triangles as shown in figure~\ref{fig:triangle_projections}.
\begin{center}
\begin{minipage}{8 cm}
\centering
\begin{pspicture}(0,0)(8,4)
\psframe[linecolor=black](0.5,0.5)(7.5,3.7)
\pspolygon[linewidth=1.5 pt](1.2,1.4)(2,3)(2.8,1.4)
\pscircle[linestyle= dashed](2,2){1}
\pspolygon[linewidth=1.5 pt](5.2,1.4)(6,3)(6.8,1.4)
\pscurve[linewidth=1.5 pt,linecolor=blue](5.2,1.4)(5.3,2.2)(6,3)
\pscurve[linewidth=1.5 pt,linecolor=blue](6,3)(6.7,2.2)(6.8,1.4)
\pscurve[linewidth=1.5 pt,linecolor=blue](6.8,1.4)(6,1.2)(5.2,1.4)
\pscurve[linewidth=1.5 pt,linecolor=red](5.2,1.4)(5.8,2.2)(6,3)
\pscurve[linewidth=1.5 pt,linecolor=red](6,3)(6.2,2.2)(6.8,1.4)
\pscurve[linewidth=1.5 pt,linecolor=red](6.8,1.4)(6,1.6)(5.2,1.4)
\pscircle[linestyle= dashed](6,2){1}
\psline[linewidth=2 pt]{->}(3.2,2)(4.8,2)
\rput(4,2.2){$\pi$}
\end{pspicture}
\end{minipage}
\caption{Curvilinear triangles as projections of rectilinear triangles }
\label{fig:triangle_projections}
\end{center}
One important generalization is that the requirement of the Delta($\Delta$)-sets to be polytopes is relaxed. This can be understood using the following example.

\begin{example}
\begin{minipage}{7 cm}
\centering
\begin{pspicture}(-0.5,-0.5)(2.2,2.2)
\psframe[linecolor=black](-0.3,-0.3)(2,1.25)
\psdots[dotstyle=o,fillstyle=solid,fillcolor=black](0.5,0.5)(1.5,0.5)
\pscurve{->}(0.5,0.5)(1,0.75)(1.5,0.5)
\pscurve{->}(0.5,0.5)(1,0.35)(1.5,0.5)
\rput(0.3,0.3){$v_1$}
\rput(1.7,0.3){$v_2$}
\rput(1,1){$e_1$}
\rput(1,0.15){$e_2$}
\end{pspicture}
\end{minipage}
\caption[]{Curved simplices in Delta($\Delta$)- sets}
\label{fig:deltaset}
Let us look at the simplices in the figure~\ref{fig:deltaset}. There are two curved lines joining the two vertices $v_1$ and $v_2$. It is important to note that unlike the ordered simplicial complexes a collection of vertices does not specify a unique simplex. We can prove that both the curved lines between the vertices form valid Delta($\Delta$)-set. The $X_0=\{v_1,v_2\}$ and the $X_1=\{e_1,e_2\}$ satisfy the face map relation $d_id_j=d_{j-1}d_i$ for $i<j$. It is obvious that $d_0d_1(e_1)=d_0d_0(e_1)=v_2$ and $d_0d_1(e_2)=d_0d_0(e_2)=v_1$. Hence both the $e_1$ and $e_2$ are valid Delta($\Delta$)-sets as per the definition~\ref{def:delta_sets}.
\qquad \textcolor{blue}{\Squaresteel}
\end{example}
These sets can now be combined, to give a notion of Delta($\Delta$)-complexes, using the notion of face maps and inclusion maps to yield complex topological structures. This can be represented using the same bundle framework depicted in figures~\ref{fig:bundle} and~\ref{fig:chain_bundle}. Now, that we have the definition of the Delta($\Delta$)-sets, let us move on to the topological structures that are built on simplicial complexes. These structure can be easily generalized to Delta($\Delta$)-sets, because these structures are built on the notion of a set.
\begin{definition}\label{def:nerve}
The nerve of a simplicial complex (or a Delta($\Delta$)-sets) $K$ denoted by $\Nrv K$, is defined as the collection of simplices with a common intersection. This is formally defined as $\Nrv K=\{\Delta \subseteq K \mid \bigcap \Delta \neq \phi \}$ 
\end{definition}
Moreover, an other concept closely related to the $\Nrv K$ is that of a spoke.
\begin{definition}\label{def:1-spoke}
The spoke ($1$-spoke) denoted by $\sk A$ is the name for each of the sets that share a common intersection. Thus, for simplicial complexes (Delta($\Delta$)-sets), each of the constituent simplices of the $\Nrv K$ is a spoke. 
\qquad \textcolor{blue}{\Squaresteel}
\end{definition}

An associated notion is that of a nucleus of the nerve, which is defined as the intersection of all the spokes in the nerve $\Nrv K$.We generalize this notion of a spoke to $k$-spokes where $k \in \mathbb{Z}$ and $k \geq 0$.

\begin{definition}\label{def:k-spoke}
A $k$-spoke denoted by $sk_k$, $k\geq 0$ and $k \in  \mathbb{Z}$ is a topological structure which generalizes the notion of a complex. A $sk_k$ in a simplicial complex $K$ is a simplex (or a Delta($\Delta$)-set) that has a non empty intersection with a simplex (or a Delta($\Delta$)-set) in the $sk_{k-1}$. This is a recursive definition with the base case $sk_0$ = nucleus. This can be formally defined as each element of the set $\{\Delta \subseteq K \backslash \{\bigcup sk_{k-1}\}\mid \Delta \bigcap \{\bigcup sk_{k-1}\} \neq \phi \}$ for $k>1$ and for $k=0$ it is equal to the nucleus.
\qquad \textcolor{blue}{\Squaresteel}
\end{definition}


%
\begin{wrapfigure}[13]{R}{0.35\textwidth}
\begin{minipage}{3 cm}
\centering
\begin{pspicture}(0,1)(5,5)
\psframe(0.6,1)(4.3,4.3)
\psccurve[fillstyle=solid,fillcolor=green](1.7,3.7)(2.2,3.6)(2.7,3.7)(2.2,3)(1.7,2.7)(1.8,3.5)
\psccurve[fillstyle=solid,fillcolor=green](3.7,3.7)(3.4,3.8)(2.7,3.7)(3.2,3)(3.7,2.7)(3.8,3.5)
\psccurve[fillstyle=solid,fillcolor=red](2.7,2.7)(2.6,3.2)(2.7,3.7)(3.2,3)(3.7,2.7)(3.2,2.6)
\psccurve[fillstyle=solid,fillcolor=red](2.7,2.7)(2.6,2.2)(2.7,1.7)(3.2,2.3)(3.7,2.7)(3.2,2.6)
\psccurve[fillstyle=solid, fillcolor=red](2.7,2.7)(2.6,2.2)(2.7,1.7)(2.2,2.3)(1.7,2.7)(2.2,2.5)
\psccurve[fillstyle=solid,fillcolor=red](1.7,2.7)(2.2,2.5)(2.7,2.7)(2.6,3.2)(2.7,3.7)(2.2,3)
\psccurve[fillstyle=solid,fillcolor=green](2.7,1.7)(3.2,1.6)(3.7,1.7)(3.8,2.5)(3.7,2.7)(3.3,2.5)
\psccurve[fillstyle=solid,fillcolor=green](2.7,1.7)(2.2,1.6)(1.7,1.7)(1.6,2.5)(1.7,2.7)(1.86,2.5)
\rput(1,4){$X$}
\psdots[dotsize=0.12,dotstyle=o,fillstyle=solid,fillcolor=black](1.7,2.7)(2.7,2.7)(3.7,2.7)(2.7,1.7)(2.7,3.7)(3.7,3.7)(3.7,1.7)(1.7,1.7)(1.7,3.7)
\psdots[dotsize=0.25,dotstyle=triangle,fillcolor=yellow](2.7,2.7)
\end{pspicture}
\end{minipage}
\caption[]{$\sk$-chains}
\label{fig:spkstructs}
\end{wrapfigure}

Another topological structure that we will build using the notion of $k$-spokes, can prove to be instrumental in the representation of shapes in images. This concept is that of a $k$-spoke complex.

\begin{definition}\label{def:spoke_complex}
A $k$-spoke complex ($\mbox{skcx}_k$) is the union of all the $k$-spokes ($sk_k$) in the image. 
\qquad \textcolor{blue}{\Squaresteel}
\end{definition}

Another topological structure that is built using the combination of simplicial complexes or Delta($\Delta$)-sets is the $k$-spoke chain. It is a chain of sets with each of the adjacent sets having a non-empty intersection. It is defined as follows:

\begin{definition}\label{def:kspoke_chain}
The $k$-spoke chain is denoted by $skchain_k$. It is defined as, $\{\bigcup_{j=0}^kA_j \in \mbox{skcx}_j\mid A_i \cap A_{i+1} \neq \phi \}$. 
\qquad \textcolor{blue}{\Squaresteel}
\end{definition}
%
\begin{example}
To understand the topological structures defined in the definitions~\ref{def:nerve} - \ref{def:kspoke_chain} let us look at the figure~\ref{fig:spkstructs}. The collection of triangles with curved edges colored red are the nerve as defined in the definition~\ref{def:nerve}. Each of the red triangles is individually a $1$-spoke as defined in the definition~\ref{def:1-spoke}. The point denoted by the yellow filled triangle at the common intersection of all the red triangles is nucleus.

From this example we can illustrate the concept of $k$-spokes. The nucleus (yellow filled triangle) is the $0$-spoke, each of the red triangles is $1$-spoke and each of the green triangles is the $2$-spoke. This can be confirmed from the definition~\ref{def:k-spoke}. Moving on to the notion of the spoke complexes we can see from the definition~\ref{def:spoke_complex} that the nucleus (yellow filled triangle) forms the $0$-spoke complex, the union of all the red triangles forms the $1$-spoke complex and the union of all the green triangles is the $2$-spoke complex.

Let us now use the same figure to explain the notion of a $k$-spoke chain. The nucleus (yellow filled triangle) $0$-spoke chain. The union of the the nucleus with any one of the red triangles, with which has non-empty intersection, is called the $1$-spoke complex. The union of any one of the $1$-spoke complexes and a $2$-spoke, with which it has non-empty intersection.
\qquad \textcolor{blue}{\Squaresteel}
\end{example}

\section{Shape Geometry and Curved Triangulations}
The study of shapes and their geometries has been an important topic in both topology\cite{BorsukDydak1980WhatIsTheTheoryOfShape} and computer vision\cite{kazmi2013survey}. Both the fields have developed techniques to cater this problem. The shape theory that was developed by Borsuk was an important contribution to the field of geometry and topology\cite{borsuk1968concerning}\cite{Borsuk1948FMsimplexes}. This field later developed into algebraic topology, which has gained recent interest in computer vision\cite{freedman2009algebraic}. Some approaches that have been developed for shape geometry using heat kernel signatures\cite{bronstein2010scale}, wave kernel signatures\cite{aubry2011wave} and Riemannian geometry based frameworks\cite{laga2012riemannian}\cite{gris2016sub}. Recently, a new framework has been proposed, that aims at capturing the shape geometry by means of triangulations or tessellation of the digital image\cite{Peters2017ComputerVision}\cite{Peters2013springer}.

\setlength{\intextsep}{0pt}
\begin{wrapfigure}[12]{R}{0.45\textwidth}
\begin{minipage}{2.8 cm}
\centering
\includegraphics[width=55mm]{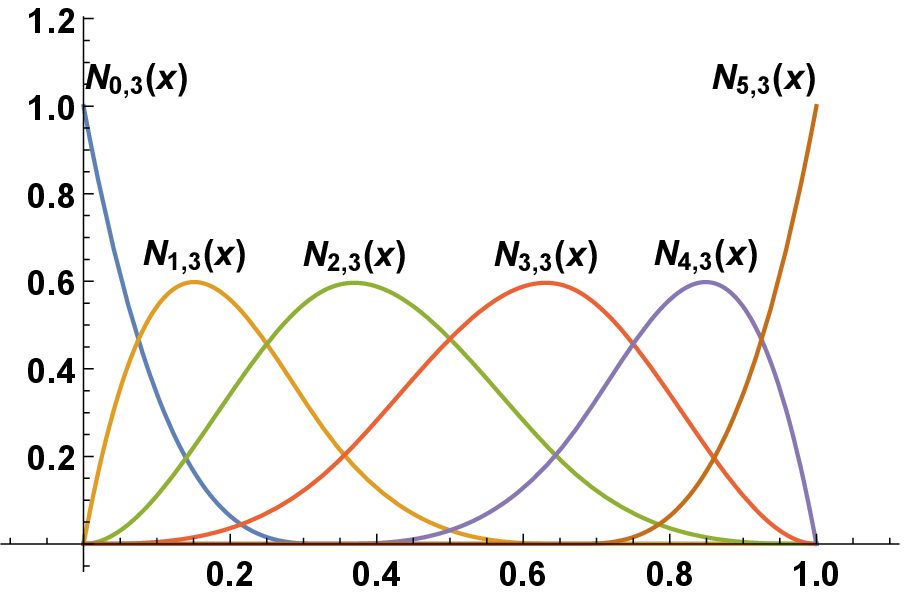}
\caption[]{\footnotesize Bspline}
\label{fig:splinebasis}
\end{minipage}
\end{wrapfigure}

To construct curved triangulations based on the keypoints identified in the image, we use B-splines to construct the triangulations instead of straight lines. The B-splines are generalizations of Bezier curves. The curvature of the line is controlled by a set of control points and knots. The locations and the weights of the control points is instrumental in determining the points that are interpolated by the given spline. The number of duplicate knots is crucial in determining the smoothness of the curve as elaborated by remark~\ref{rem:differentiability}. The distance between the non-duplicate knots is the range of each of the basis functions of a specific degree. As represented in the definition below, the resulting curve is the linear combination of basis functions (an example of B-spline with basis of degree r is shown in figure~\ref{fig:splinebasis}).
     
\begin{definition}\label{def:bspline}
Let a $m+1$-dimensional vector be defined as $T=\{t_0,t_1,\cdots,t_m\}$, where $T$ is a non-decreasing sequence with $t_i \in [0,1]$. The vector $T$ is called a knot vector. A set of control points are also defined as $P_0,\cdots,P_n$. The degree of the spline is defined as $p=m-n-1.$ The basis functions are defined as:

%

\begin{align*}
N_{i,0}(t) &= \begin{cases}
    1 & \text{if $t_i \leq t <t_{i+1}, t_i <t_{i+1}$,}\\
    0 & \text{otherwise.}
   \end{cases}\\
N_{i,j}(t) &= \frac{t-t_i}{t_{i+1}-t_i}N_{i,j-1}(t)+\frac{t_{i+j+1}-t}{t_{i+j+1}-t_{i+1}}N_{i+1,j-1},\mbox{where}\\
           & j=1,2,\cdots, p.
\end{align*}					
The curve defined by $C(t)=\sum_{i=0}^{n}P_iN_{i,p}(t)$ is a B-spline.
\qquad \textcolor{blue}{\Squaresteel}
\end{definition}

The knot vector $T$ plays an important role in determining the smoothness properties of the resulting B-spline.

\begin{remark}\label{rem:differentiability}
The B-spline curve is $p-k$ times differentiable at the point where there are $k$ duplicate knots.
\qquad \textcolor{blue}{\Squaresteel} 
\end{remark}

\begin{wrapfigure}{R}{0.5\textwidth}
\begin{minipage}{6cm}
\centering
\begin{pspicture}(0,0)(3.5,3.5)
\psframe[linecolor=black](0.3,0.3)(3.4,3.4)
\psdots[dotsize=0.2,dotstyle=o,fillcolor=red](2,2)(1,1.5)(3,1)(1.9,0.7)(2.8,2.3)(0.9,2.3)
\pspolygon[linecolor=black](2,2)(1,1.5)(3,1)
\pspolygon[linecolor=black,linestyle=dashed](2,2)(1,1.5)(0.9,2.3)
\pspolygon[linecolor=black,linestyle=dashed](2,2)(3,1)(2.8,2.3)
\pspolygon[linecolor=black,linestyle=dashed](3,1)(1,1.5)(1.9,0.7)
\psBspline[linecolor=blue](1,1.5)(2,1.4)(1.9,1)(3,1)
\psBspline[linecolor=blue](1,1.5)(1.6,1.7)(2,1.5)(2,2)
\psBspline[linecolor=blue](2,2)(2.6,1.9)(1.7,1.7)(3,1)
\end{pspicture}
\end{minipage}
\caption[]{Projecting rectilinear triangulation to curvilinear triangulation based on keypoint location}
\label{fig:curv_triang}
\end{wrapfigure}

The proposed approach builds curved triangulations in a simplicial complex derived from the location of  keypoints in a local region. It can be easily verified from basic knowledge of geometry, that an edge in a triangulation can be shared by at least $1$ and at most $2$ triangles in a complex. If an edge is a part of a triangle, we have three control points, so that a triangle edge is defined by a pair of the points and the third point controls the curvature. If an edge is shared by two triangles, then apart from the two end points, there are two points that control the curvature giving us a total of four control points. Using this approach, we project the whole Delaunay (rectilinear) triangulation to a curvilinear (B-Spline) triangulation, an edge at a time.This is illustrated in figure~\ref{fig:curv_triang}. Assigning weights to the control points allows for controlling the tension in the spline. If the weights of the endpoints is increased the spline has more tension and its looping behavior is reduced. Thus modifying the weights yields different triangulations. Keeping Figure~\ref{fig:triangle_projections} in mind, we can represent this triangulation method algorithm as a fiber bundle structure $\pi:\mathcal{E} \rightarrow \mathcal{B}$, where $\mathcal{E}$ is the rectilinear triangulation space and the $\mathcal{B}$ is the curvilinear triangulation space.

The shape geometry proposed in \cite{peters2017proximal} takes into account the topological structures defined in the definitions~\ref{def:nerve}-~\ref{def:kspoke_chain}, extracted from a digital image following the steps in Alg.~\ref{alg:ImageGeometry}.

\begin{algorithm}[!ht]
\caption{Rectilinear Triangulations of Digital Images}\label{alg:RectTriang}

\SetKwData{Left}{left}
\SetKwData{This}{this}
\SetKwData{Up}{up}
\SetKwFunction{Union}{Union}
\SetKwFunction{FindCompress}{FindCompress}
\SetKwInOut{Input}{Input}
\SetKwInOut{Output}{Output}
\SetKwComment{tcc}{/*}{*/}

\Input{digital image $img$}
\Output{Rectilinear Triangulation Mesh $\mathscr{M}$ covering an image}
\emph{$img \longmapsto SelectedVertices$}\;
\emph{$S \leftarrow SelectedVertices$}\;
/* $S$ contains vertex coordinates used to triangulate $img$. */ \;
\emph{Continue  $\leftarrow$ True; Vertices  $\leftarrow$ emptyset}\;
\While {($S\neq\emptyset\ \&\ Continue = True$)}{
  /* Select neighbouring vertices $\left\{p,q,r\right\}\subset S$. */ \;
  \emph{$Vertices \assign Vertices\cup \left\{p,q,r\right\}$}\;
	/* $\fil \Delta(pqr)$ = intersection of three closed half spaces. */ \;
	\emph{$\mathscr{M}\ \assign\ \mathscr{M}\cup \fil \Delta(pqr)$}\;
	\emph{$S\ \assign\ S\setminus Vertices$}\;
	\eIf{$S\neq \emptyset$}{
	/* Continue */ }{
	\emph{Continue $\leftarrow$ False}\;
	} 
} 
\emph{$\mathscr{M}\longmapsto\ img$} \;
\end{algorithm}

Alg.~\ref{alg:ImageGeometry} takes the rectilinear mesh generated by Alg.~\ref{alg:RectTriang} as an input and outputs the B-spline based curvilinear mesh. It is  obvious from the discussion above that the shape geometry will be dependent on the keypoints. The basic assumption behind this framework is, that an object to be detected is the most distinctive feature of the image. An image object is usually a region with relatively similar (proximal) pixel values. Thus, the boundary of an object can be identified from the gradient of the pixel values. Building on this observation, in the current investigation we employ gradient-based keypoints\cite{Lowe2004IJCVSIFTkeypoints}.\\
\vspace{3mm}

\begin{algorithm}[!ht]
\caption{Digital Image Geometry via Curvilinear Triangulations}\label{alg:ImageGeometry}

\SetKwData{Left}{left}
\SetKwData{This}{this}
\SetKwData{Up}{up}
\SetKwFunction{Union}{Union}
\SetKwFunction{FindCompress}{FindCompress}
\SetKwInOut{Input}{Input}
\SetKwInOut{Output}{Output}
\SetKwComment{tcc}{/*}{*/}

\Input{digital image $img$, Rectilinear Mesh $\mathscr{M}$,spline weight vector $w$,
       from Alg.~\ref{alg:RectTriang}}
\Output{Curvilinear Triangulation Mesh $\mathscr{O}$ covering an image region}
\emph{$T \leftarrow edges(\mathscr{M})$}\;
/*$T$ contains all the edges in the triangulation $\mathscr{M}$*/ \;
\ForEach{edge $\in T$}{
$tri \leftarrow triangles adjacent to the edge$\;
/*Extract all the unique vertices in $tri$*/\;
$vert \leftarrow \bigcup_{tri}vertices(tri)$\;
/*arrange(vert), arranges the vertices to place the endpoint of the edge as the first and the last element respectively*/\;
$vert \leftarrow arrange(vert)$\;
/*Bspline(control-points,weights) generates a Bspline. The Knot locations are uniformly spaced and chosen to ensure that spline connects the end points*/\;
$\mathscr{O} \assign \mathscr{O} \cup Bspline(vert,w)$\;
}
$\mathscr{O} \longmapsto img$\;
/* Use $\mathscr{O}$ to gain information about image shape geometry. */ \;
\end{algorithm}
$\mbox{}$\\
The shape geometry proposed in \cite{peters2017proximal} takes into account the topological structures defined in the definitions~\ref{def:nerve} to \ref{def:kspoke_chain}, extracted from a digital image following the steps in Alg.~\ref{alg:ImageGeometry}. The Alg.~\ref{alg:ImageGeometry} takes the rectilinear mesh generated by Alg.~\ref{alg:RectTriang} as an input and outputs the B-spline based curvilinear mesh.  This shape geometry is dependent on the keypoints used in the triangulation. 
The basic assumption underlying this framework is that an object to be detected is the most distinctive feature of the image.  
An image object is usually a region with relatively similar (proximal) pixel values. Thus, the boundary of an object can be identified from the gradient of the pixel values. Building on this observation, we employ gradient-based keypoints\cite{Lowe2004IJCVSIFTkeypoints}. 

Keypoints are prone to lie on regions in a image that have high gradient contrast, {\em i.e.}, keypoints lie on the boundary of an image object. Thus, when we triangulate using keypoints, there is a fair chance that a nerve complex of this triangulation lies on the interior of the shape. A \emph{\bf nerve complex} is a collection of simplexes that have nonempty intersection.    For example, in a triangulation, a nerve complex is a collection of 2-simplexes that have a vertex in common. A question arises here. With multiple nerves in a triangulation, which nerve complex do we choose? Since, we had made the assumption that the object that is under analysis is the most distinctive feature of the image, we look for the maximal nerve cluster ($\maxNrvClu$). The maximal nerve cluster is the nerve that results from the intersection of maximum number of sets in the image. This is because most of the keypoints in the image will be clustered on the boundary of the object. As a result, most of triangles with these points as vertices, would lie on the interior of the shape. Since most of the triangles lie on the shape it can be concluded that the nerve of these triangles would be the $\maxNrvClu$. 

Once we conclude that the maximal nerve cluster($\maxNrvClu$) is on the interior of the shape we can start to model the shape as a union of $k$-spoke complexes($skcx_k$) or the $k$-spoke chains($skchain_k$). This allows us to formulate a topology(\cite{munkres1984elements}) over the shape with the triangles in the triangulation as the constituent sets. Hence we can refer to the triangulated shape as a topological space $\mathscr{S}$ This gives us a framework to discuss both the local and the global features of the space $\mathscr{S}$. The global features such as homology and homotopy are the subject of topology. The local features limited to the constituent sets such as the area, parameter, centroid and the curvature of the edges of the triangles are the subject of geometry. Moreover, the features pertaining to the pixel values of the constituent sets can be termed as their description. Other features can be adopted from the computer vision and machine learning literature to incorporate into this framework. The merger of the local and the global features allow for a comprehensive representation of $\mathscr{S}$. This representation of shape will be detailed in the following section.    
\section{Applications in Shape Analysis}
In this section we will compare rectilinear and curvilinear triangulations from the perspective of shape analysis in images. 
\begin{figure}[!ht]
\centering
\subfigure[Delaunay]{
 \includegraphics[width=55mm]{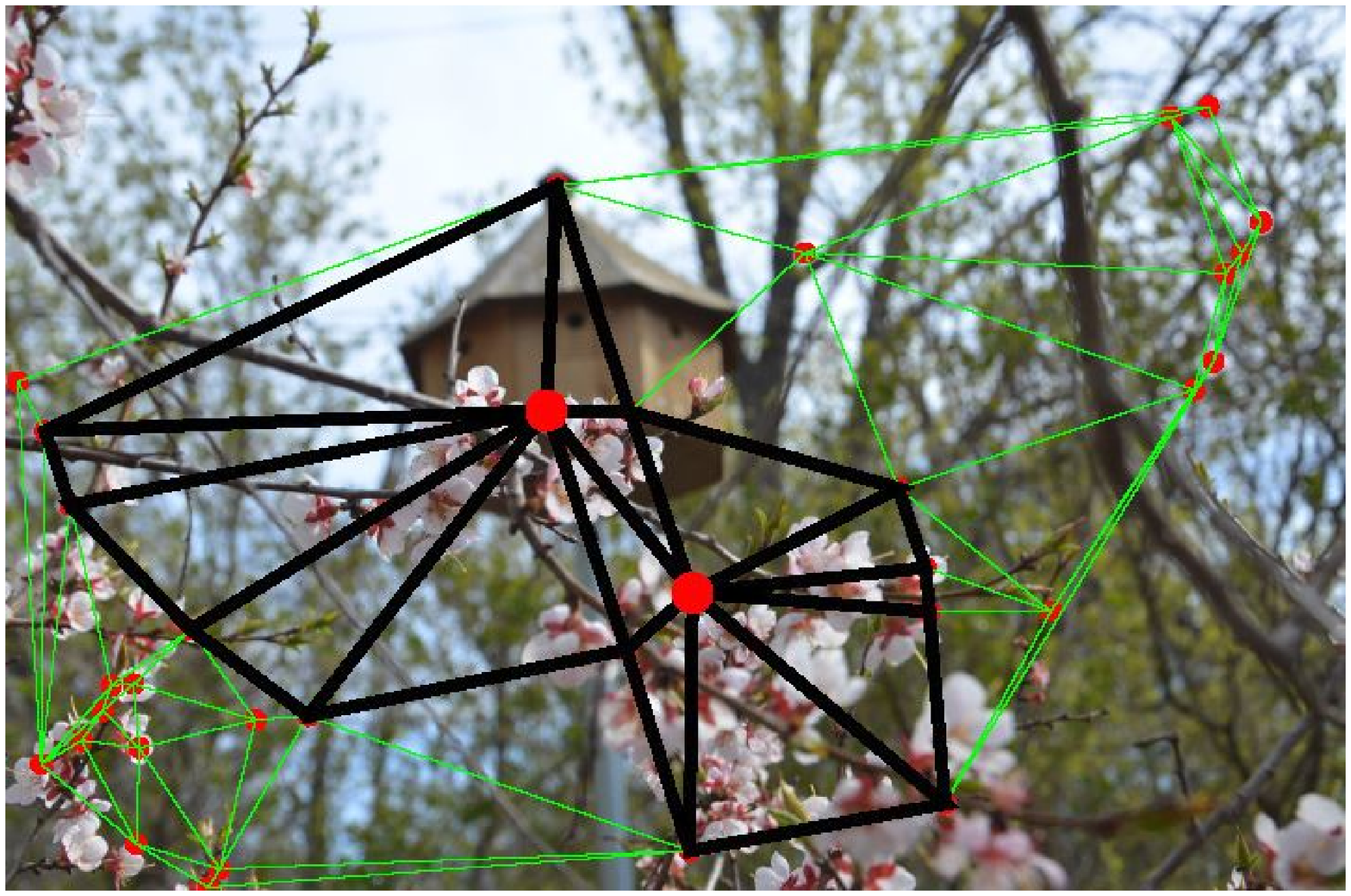}\label{subfig:lilac_shape_del}}\hfil
\subfigure[Bspline]{
 \includegraphics[width=55mm]{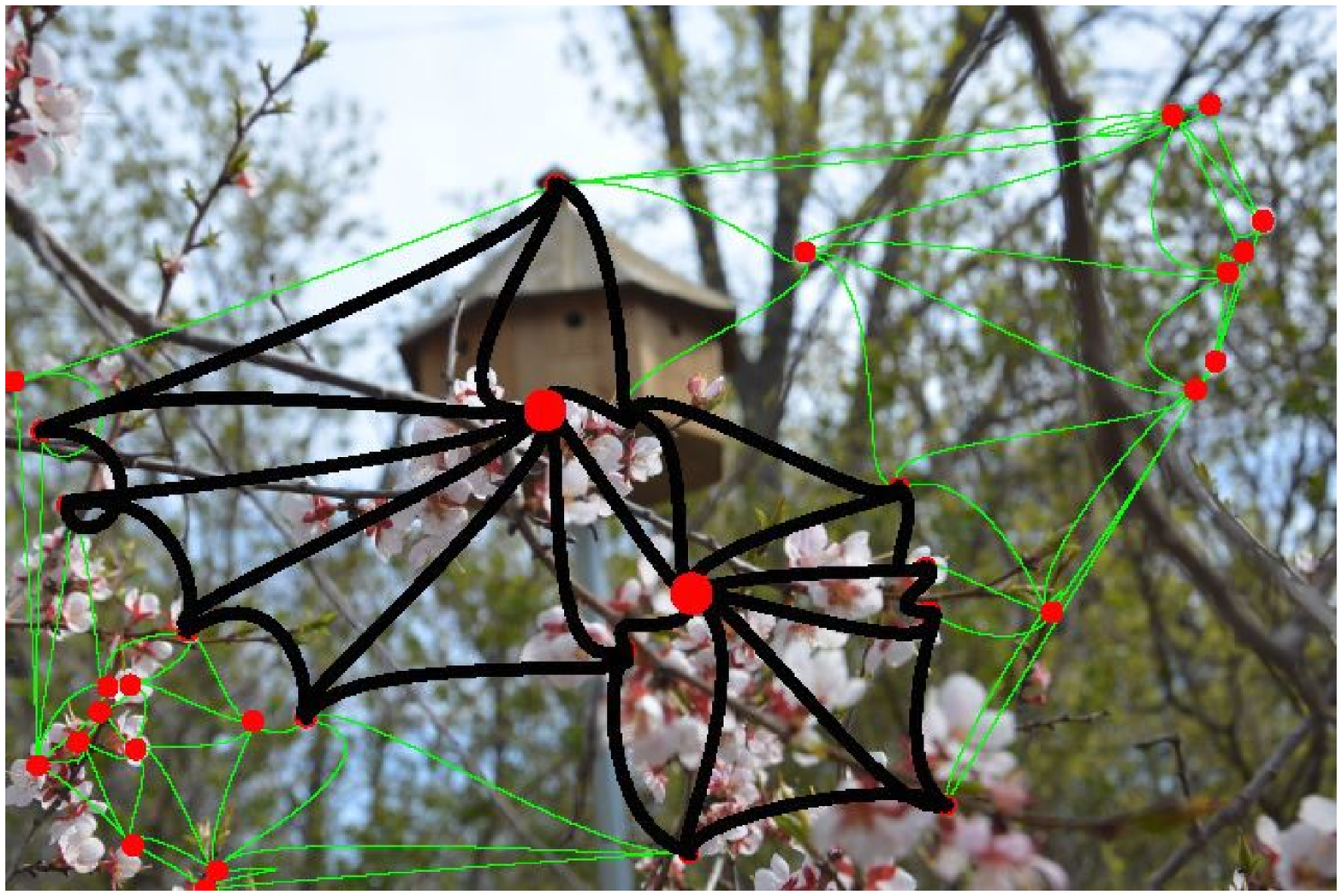}\label{subfig:lilac_shape_bspline}}
\caption{Estimating the shape geometry of lilac flower bunch}
\label{fig:lilac_shape}
\end{figure}

\begin{figure}[!ht]
\centering
\subfigure[Delaunay]{
 \includegraphics[width=55mm]{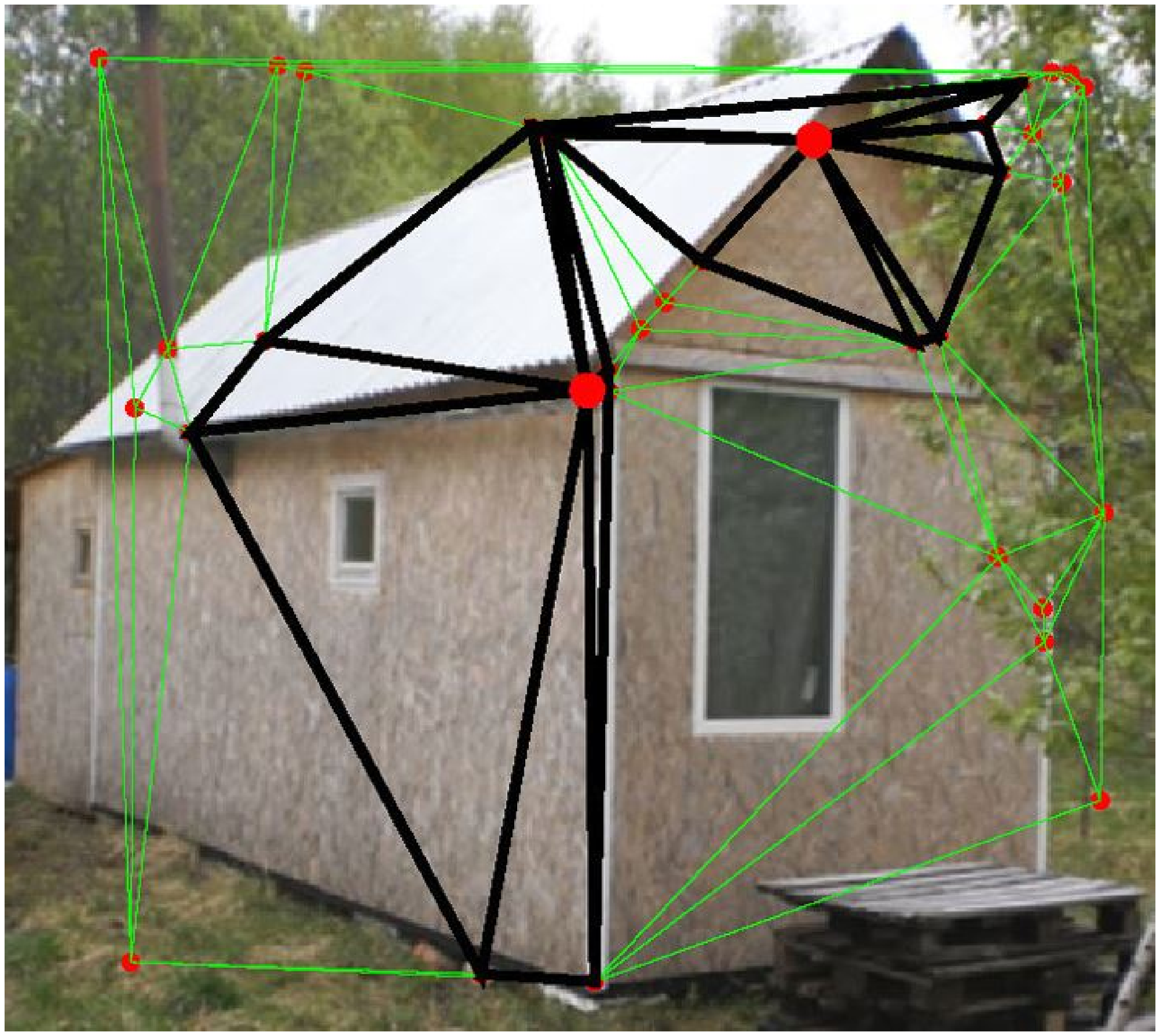}\label{subfig:cottage_shape_del}}\hfil
\subfigure[Bspline]{
 \includegraphics[width=55mm]{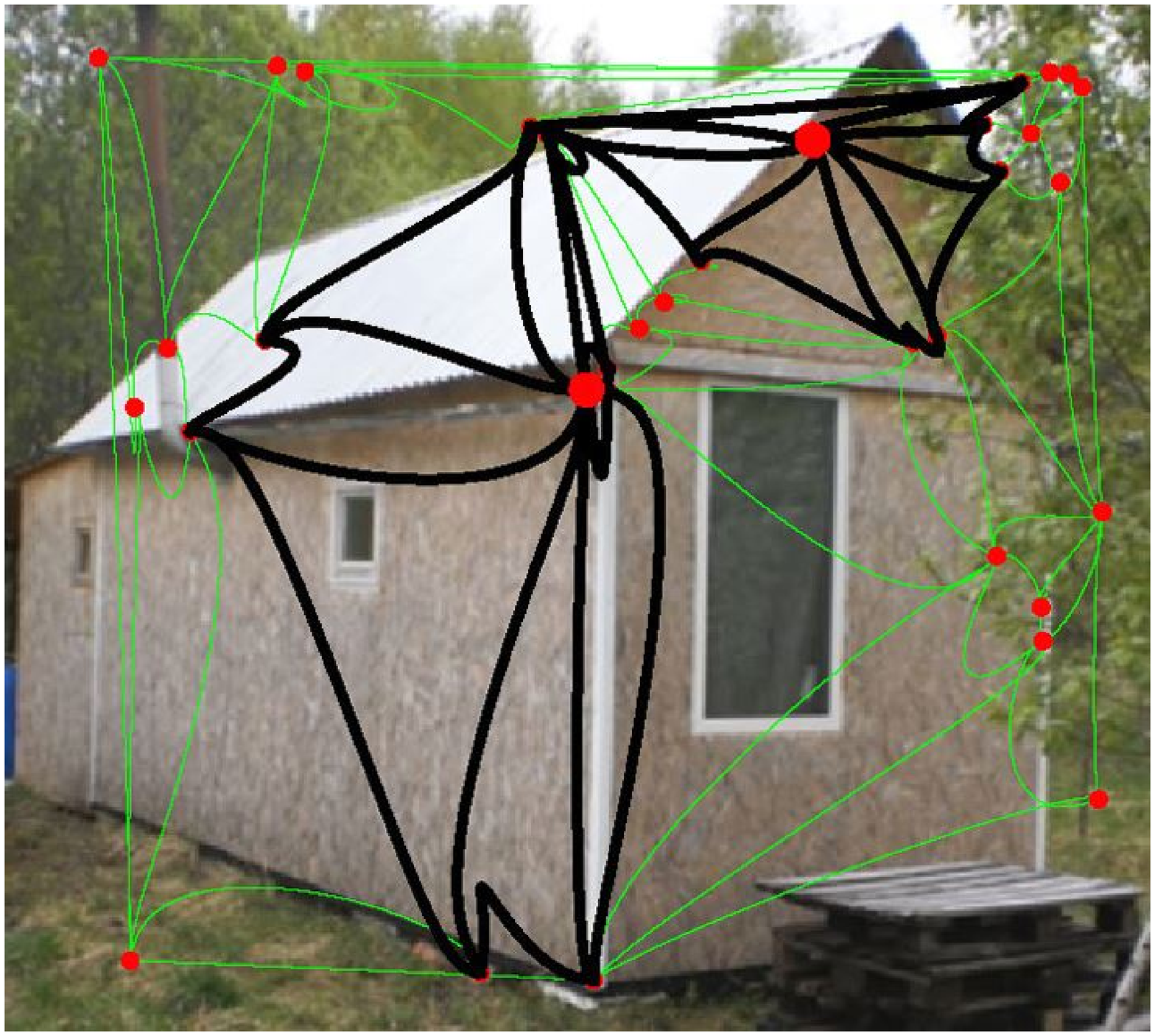}\label{subfig:cottage_shape_bspline}}
\caption{Estimating the shape geometry of a cottage}
\label{fig:cottage_shape}
\end{figure}
At first, let us compare the results for both the triangulations for the image with curved object boundaries. 
The shape of the bunch of flowers shown in figure~\ref{fig:lilac_shape} has curved boundaries. It can be observed (figure~\ref{subfig:lilac_shape_bspline}) that the curvilinear triangulations, especially the $1$-spoke complex triangles on the central flower bunch, curve with the boundary between the flowers and the background. The rectilinear triangulation(figure~\ref{subfig:lilac_shape_del})does not conform to the boundary of the flower bunch in the same manner as the curviliner triangulations. The edges of the triangulations lie on the flower buds themselves. From these two figures we can see that for the objects with curved boundaries curvilinear triangulations are better at revealing the true object shape as compared to the rectilinear triangulations.


  %
%
To present a complete analysis let us compare both forms of triangulations for an image with objects having straight boundaries(figure~\ref{fig:cottage_shape})\footnote{Many thanks to Alexander Yurkin for this image that shows a cottage that he designed.}. 
For this image, it can be observed that the rectilinear triangulations(figure~\ref{subfig:cottage_shape_del}) conform to the boundaries of the object more than the curvilinear triangulations(figure~\ref{subfig:cottage_shape_bspline}). Even though the curvilinear triangulations do not conform to the boundary of the cottage in the image, the triangulations do not escape the interior of the shape more than the rectilinear triangulations. This fact can be observed by looking at the $1$-spoke complex triangles depicted with black edges for both the rectilinear and the curvilinear triangulations presented in figure~\ref{fig:cottage_shape}.
From the results(figures~\ref{fig:lilac_shape} and \ref{fig:cottage_shape}) discussed above it can be concluded that the curvilinear triangulations are better at revealing the shape geometry of objects with curved boundaries while the rectilinear triangulations are better at conforming to the shape geometries with straight edges. An important observation is that the curvilinear triangulation does not spill outside the rectilinear triangulation. This result considers the perimeter of the triangulation and does not talk about the individual triangles. 
This is a direct result of the way the curved triangles are designed. Since for the edges that line the perimeter of the curvilinear triangulation have control points that lie on the interior of the convex polygon defined by the keypoints on the perimeter. This dictates that the B-spline segments that form the perimeter will curve towards the control points that lie on the interior of the triangulation. This leads to the fact that the curvilinear triangulation will not spill outside the rectilinear triangulation. This is both a pro and con of the proposed curvilinear triangulation, since it is unable to conform to image object shape boundaries exactly. For example, in the images in figure~\ref{fig:lilac_shape} and figture \ref{fig:cottage_shape}, although the B-spline triangulations tend to conform to the convex shape boundaries, both forms of triangulations spill over shape boundaries. This error in the approximation can be reduced by increasing the number of keypoints on the boundary. The best approximation occurs when all the keypoints lie on an image shape, either in the interior or on the shape boundary.      

\section{Results for Proximal Complexes}

In this section, we will introduce object spaces $\mathscr{O}$ and equip them with proximity relations. We will equip this object space with proximity relations to yield a proximal relator space, later on in this section. Before that, let us mention the different proximity relations that we are going to use. These include the Lodato proximity($\near$), strong proximity($\sn$) and the descriptive proximity($\dnear$), which have been discussed with detail in~\cite{peters2017proximal}.
Let us define the strong proximity($\sn$):

\begin{definition}\label{def:sn}
Let $A$ and $B$ be two nonempty sets in a topological space $X$ equipped with a proximity $\sn$. Then $A \cap B \Rightarrow A \; \sn \; B$, meaning $A$ and $B$ are strongly near.
\qquad \textcolor{blue}{\Squaresteel}
\end{definition}

 Moreover, it must be noted, that the nerve and all the spoke structures are defined with respect to a nucleus, {\em i.e.}, a vertex $p$. The nerve complex(\Nrv K(p)), $k$-spoke($sk_k K(p)$) and the $k$-spoke complex($skcx_k K(p)$) where $K$ is a simplicial(or ordered simplicial) complex and p is nucleus. 

As discussed above the nucleus of the maximal nerve cluster is assumed to lie on the interior of the object. Thus, the nucleus of the maximal nerve cluster is also considered to be the nucleus of the object in the digital image. Next, we use the notion of $k$-spoke complexes(definition~\ref{def:spoke_complex}) to formalize the notion of an object space.

\begin{definition}\label{def:object_space}
The object space, $\mathscr{O}_p$, is defined as the union of all the k-spoke complexes with vertex p as the nucleus, {\em i.e.}, $\mathscr{O}_p \assign \bigcup_k skcx_k K(p)$.
\qquad \textcolor{blue}{\Squaresteel}
\end{definition} 

Let, $A$ be a nonempty set in a topological space $X$.  The \emph{boundary} of $A$ (denoted by $\bdy A$) is set of all points that are close to $A$ and close to the complement of $A$.  The closure of $A$ (denoted by $\cl A$) is the set of all vertices $x\in X$ such that $s\ \delta\ A$.   Recall that the interior of a nonempty set $A$ (denoted by $\Int A$) is defined by $\Int A = \cl A - \bdy A$.  Every nonempty set $A$ with a boundary and with a nonempty interior defines a shape (denoted by $\shape A$).  Assuming, that the nucleus of the maximal nerve cluster is on the interior of the shape we prove the nerve of the nucleus is strongly near the interior of the shape. 

\begin{theorem}\label{thm:theorem1}
A vertex $p \in \Int(\shape A) \Rightarrow \Int(\shape A)\; \sn \; \Nrv K(p)$.
\end{theorem}
\begin{proof}$\mbox{}$\\
The vertex $p \in \Nrv K(p)$ by definition~\ref{def:nerve}, as it is the nucleus of the nerve. If $p \in int(\shape A)$, then this implies that vertex $p \in \{\Nrv K(p) \cap \Int(\shape A)\}$. This, by definition~\ref{def:sn}, means that $int(\shape A) \; \sn \; \Nrv K(p)$.
\end{proof}

Since, we consider proximity spaces over simplicial(or ordered simplicial) complex, so we can assign a gradation to the notion of strong proximity.

\begin{definition}\label{def:sn_k}
Let $sxA$ and $sxB$ be two simplicial (or ordered simplicial) complexes in a planar region, such that $sxA \; \sn  \; sxB$, then $sxA \cap sxB \neq \emptyset$ will be a simplicial(or ordered simplicial complex) of dimension $k=0,1, 2$. We can say that $sxA \; \sn_k \; sxB$, {\em i.e.}, $sxA$ is $k$-strongly near $sxB$.
\qquad \textcolor{blue}{\Squaresteel}
\end{definition}

Thus, we extend the qualitative notion of strong proximity to a quantitative notion of graded strong proximity $sn_k$.
Next, consider the hierarchy of graded strong proximities.

\begin{theorem}\label{thm:theorem2}
Let $sxA$ and $sxB$ be two simplicial(or ordered simplicial) complexes, then $sxA \; \sn_k \; sxB \Rightarrow sxA \; \sn_{k-1} \;  sxB,\, k \in \mathbb{Z}^+$.
\end{theorem}
\begin{proof}$\mbox{}$\\
$sxA \; \sn_k sxB$ implies that $sxA \cap sxB$ is a $k$-dimensional simplicial(or ordered simplicial) complex, which by definition~\ref{def:simplicial_complex} contains a $(k-1)$-dimensional simplicial(or ordered simplicial) complex . This implies $sxA \; \sn_{k-1} \; sx B$.
\end{proof}

An important feature of the
image geometry algorithm (Alg.~\ref{alg:ImageGeometry}), is the occurrence of multiple maximal nerve clusters. This leads to the proximity of multiple objects based on these clusters. We formulate the theorem for two objects, which can be easily generalized. 

\begin{theorem}\label{thm:theorem3}
Let $\mathscr{O}_p$ and $\mathscr{O}_{\acute{p}}$ be two different object spaces with $p$ and $\acute{p}$ as nuclei respectively in the simplicial (or ordered simplicial) complex $K$.  Then, $\mathscr{O}_p \;  \sn \; \mathscr{O}_{\acute{p}}$  if and only if $skcx_k K(p) \; \sn \; skcx_{\acute{k}} K(q)$, where $k,\acute{k} \in \mathbb{Z}^+ $.
\end{theorem}
\begin{proof}$\mbox{}$\\
$\Rightarrow$: From definition~\ref{def:object_space}, $\mathscr{O}_p \assign \bigcup_k skcx_k K(p)$, and $\mathscr{O}_{\acute{p}} \assign \bigcup_k skcx_k K(\acute{p})$.$\mathscr{O}_p \; \sn \; \mathscr{O}_{\acute{p}}$ implies that $\{\mathscr{O}_p \cup \mathscr{O}_{\acute{p}}\}= C \neq \phi$. Thus, $C \in skcx_k K(p)$ and $C \in skcx_{\acute{k} K(\acute{p})}$ for some $k,\acute{k} \in \mathbb{Z}^+$. This means that $\{skcx_k K(p) \cup skcx_{\acute{k}} K({\acute{p}})\} \neq \phi$, which implies that $skcx_k K(p) \; \sn \; skcx_{\acute{k}} K({\acute{p}})$. \\
$\Leftarrow$:
Using a similar argument,$skcx_k K(p) \; \sn \; skcx_{\acute{k}} K(\acute{p})$ implies $\{skcx_k K(p) \cup skcx_{\acute{k}} K(\acute{p})\}=D \neq \phi$. This means that $D \in skcx_k K(p)$  and $D \in skcx_{\acute{k}} K(\acute{p})$ for some $k,\acute{k} \in  \mathbb{Z}^+$. This from the definition~\ref{def:object_space} means that $D \in \mathscr{O}_p$ and $D \in \mathscr{O}_{\acute{p}}$, which implies that $\{\mathscr{O}_p \cup \mathscr{O}_{\acute{p}}\} \neq \phi$. Hence, from definition~\ref{def:sn}, $\mathscr{O}_p \;  \sn \; \mathscr{O}_{\acute{p}}$.
\end{proof}

Up to this point, a whole image is taken as a proximity space which can contain multiple objects.  Each image object is itself a subspace, denoted by $\mathscr{O}_p$, which can be equipped with proximity relations. 
\begin{lemma}\label{lm:lemma1}
Let $\left(\mathscr{O}_p, \left\{\near,\sn,\dnear\right\}\right)$  be a proximal relator space, and $skcx_k K(p) \in \mathscr{O} _p$, where $k \in \mathbb{Z}^+$.  Assuming $skcx_a K(p) \; \sn \; skcx_b K(p)$, then
\begin{compactenum}[1$^o$]
\item $skcx_a K(p) \; \sn \; skcx_b K(p) \Rightarrow skcx_a K(p) \; \near \; skcx_b K(p)$,
\item $skcx_a K(p) \, \sn \,  skcx_b K(p) \Rightarrow skcx_a K(p) \; \dnear \;  skcx_b K(p)$.
\end{compactenum}
\end{lemma}
\begin{proof}$\mbox{}$\\
\begin{compactenum}[1$^o$]
\item $skcx_a K(p) \; \sn \; skcx_b K(p)$ implies $\{skcx_a K(p) \cap skcx_b K(p)\} \neq \phi$ which in turn from the axiom (\textbf{(P3)} in \cite[\S~2.3 p.~5]{peters2017proximal}) implies $skcx_a K(p) \; \near \; skcx_b K(p)$.
\item $skcx_a K(p) \; \sn \; skcx_b K(p)$ implies $\{skcx_a K(p) \cap skcx_b K(p)\}=C \neq \phi$. This implies that $C \in skcx_a K(p)$ and  $C \in skcx_b K(p)$, thus $C \in \{skcx_a K(p) \cup skcx_b K(p)\}$. Moreover, if $\Phi(A)=\{\Phi(x) \in \mathbb{R}^n : x \in A\}$ are a set of feature vectors then it is obvious that $\Phi(C) \in \Phi(skcx_a K(p))$ and $\Phi(C) \in \Phi(skcx_b K(p))$. This, from definition(\cite[\S~4.3,p.~84]{ Naimpally2013}), implies that $A \cap_\Phi B$. Hence, from axiom(\textbf{(dP2)} in \cite[\S~2.4, p.~7]{peters2017proximal}), $skcx_a K(p) \; \dnear \;  skcx_b K(p)$.
\end{compactenum}
\end{proof}

Let us formulate the proximity relations between $k$-spoke complexes($skcx_k K(p)$) for different values of $k$.

\begin{theorem}\label{thm:theorem4}
Let $\mathscr{O}_p$ be the object space and the $skxc_k K(p) \in \mathscr{O}_p$ where $k \in \mathbb{Z}^+$. Assume that $skcx_{\hat{k}}(p)$ where $\hat{k}$ is the maximal $k$ such that $skcx_{\hat{k}}(p) \cap \bdy(\mathscr{O}_p) \neq \Phi$. Then, for $0 < j < \hat{k}$,
\begin{compactenum}[1$^o$]
\item $skcx_{j+1} K(p)\; \sn \; skcx_j K(p)$.
\item $skcx_j K(p) \; \sn \; skcx_{j-1} K(p)$.
\end{compactenum}
\end{theorem}
\begin{proof}
\begin{compactenum}[1$^o$]
\item From definition~\ref{def:spoke_complex} $skcx_{j+1} K(p)= \bigcup sk_{j+1} K(p)$  and $skcx_{j} K(p)= \bigcup sk_{k} K(p)$. It is obvious from definition~\ref{def:k-spoke} that, for all $sk_{j+1} K(p)$, there exists a $sk_j K(p)$ such that $\{sk_{j+1} K(p) \cap sk_j K(p)\} \neq \phi$. This implies that $skcx_{j+1} K(p) \cap skcx_{j} K(p)\neq \emptyset$. Hence, from definition~\ref{def:sn}, $skcx_{j+1} K(p) \; \sn \; skcx_{j} K(p)$.

\item 
Replacing $j+1$ and $j$ with $j$ and $j-1$ in the proof of 1$^o$, we obtain\\ $skcx_j K(p) \; \sn \; skcx_{j-1} K(p)$. 
\end{compactenum}
\end{proof}

Using the lemma~\ref{lm:lemma1} and the theorem~\ref{thm:theorem4}, we can prove key proximity relations for the object space $\mathscr{O}_p$.

\begin{theorem}\label{thm:theorem5}
Let $\big(\mathscr{O}_p,{\near,\sn,\dnear}\big)$  be a proximal relator space, and the $skxc_k K(p) \in \mathscr{O}_p$ where $k \in \mathbb{Z}^+$. Assume that $skcx_{\hat{k}}(p)$ where $\hat{k}$ is the maximal $k$ such that $skcx_{\hat{k}}(p) \cap \bdy(\mathscr{O}_p) \neq \Phi$. Then, for $0 < j < \hat{k}$,
\begin{compactenum}[1$^o$]
\item $skcx_{j+1} K(p) \; \sn \; skcx_j K(p) \Rightarrow skcx_{j+1} K(p) \; \near \; skcx_j K(p)$.
\item $skcx_j K(p) \; \sn \; skcx_{j-1} K(p)  \Rightarrow skcx_j K(p) \; \near \; skcx_{j-1} K(p)$.
\item $skcx_{j+1} K(p) \; \sn \; skcx_j K(p) \Rightarrow skcx_{j+1} K(p) \; \dnear \; skcx_j K(p)$.
\item $skcx_j K(p) \; \sn \; skcx_{j-1} K(p)  \Rightarrow skcx_j K(p) \; \dnear \; skcx_{j-1} K(p)$.
\end{compactenum}
\end{theorem}
\begin{proof}
\begin{compactenum}[1$^o$]
\item From theorem~\ref{thm:theorem4} we have $skcx_{j+1} K(p) \; \sn \; skcx_j K(p)$, which from lemma~\ref{lm:lemma1} gives us $skcx_{j+1} K(p) \; \near \; skcx_j K(p)$.
\item From theorem~\ref{thm:theorem4} we have $skcx_j K(p) \; \sn \; skcx_{j-1} K(p)$, which from lemma~\ref{lm:lemma1} gives us $skcx_j K(p) \; \near \; skcx_{j-1} K(p)$.
\item From theorem~\ref{thm:theorem4} we have $skcx_{j+1} K(p) \; \sn \; skcx_j K(p)$, which from lemma~\ref{lm:lemma1} gives us $skcx_{j+1} K(p) \; \dnear \; skcx_j K(p)$.
\item From theorem~\ref{thm:theorem4} we have $skcx_j K(p) \; \sn \; skcx_{j-1} K(p)$, which from lemma~\ref{lm:lemma1} gives us $skcx_j K(p) \; \dnear \; skcx_{j-1} K(p)$.
\end{compactenum}
\end{proof}
 
\section{Concluding Remarks}
This paper presents a framework for the study of shape geometry that extends the notion of rectilinear triangulations to curvilinear triangulations. This framework uses nerve and spoke structures to fuse the geometric, topological and descriptive information to study image object shapes in digital images. From the results, it can be concluded that curvilinear triangulations conform to curved objects in images better than rectilinear triangulations. Rectilinear triangulation conforms better to the geometry of objects with straight edges.  Importantly, a $k$-spoke complex is a structure which can be used to construct a layered representation of image object shapes. 



\bibliographystyle{amsplain}
\bibliography{NSrefs}

\end{document}